\documentclass[11pt]{scrartcl}

\usepackage{amsmath}
\usepackage{amssymb}
\usepackage{amsthm}

\newtheorem{theorem}{Theorem}

\newtheorem{proposition}{Proposition}
\newtheorem{corollary}{Corollary}

\newtheorem{definition}{Definition}
\newtheorem{example}{Example}

\title{Fairly Allocating Contiguous Blocks of Indivisible Items\footnote{A preliminary version of this paper appeared in Proceedings of the 10th International Symposium on Algorithmic Game Theory, September 2017.}}
\author{
Warut Suksompong\\Stanford University
}
\date{\vspace{-5ex}}

\begin{document}

\maketitle

\begin{abstract}
In this paper, we study the classic problem of fairly allocating indivisible items with the extra feature that the items lie on a line. Our goal is to find a fair allocation that is \emph{contiguous}, meaning that the bundle of each agent forms a contiguous block on the line. While allocations satisfying the classical fairness notions of proportionality, envy-freeness, and equitability are not guaranteed to exist even without the contiguity requirement, we show the existence of contiguous allocations satisfying approximate versions of these notions that do not degrade as the number of agents or items increases. We also study the efficiency loss of contiguous allocations due to fairness constraints.
\end{abstract}

\section{Introduction}

We consider the classic problem in economics of \emph{fair division}: How can we divide a set of resources among interested agents in such a way that the resulting division is fair?  This is an important issue that occurs in a variety of situations, including students splitting the rent of an apartment, couples dividing their properties after a divorce, and countries staking claims in disputed territory. The fair division literature often distinguishes between two types of resources. Some resources, such as cake and land, are said to be \emph{divisible} since they can be split arbitrarily among agents. Other resources, like houses and cars, are \emph{indivisible}---each house or car must be allocated as a whole to one agent.

To reason about fairness, we must define what it means for an allocation of resources to be fair. Several notions of fairness have been proposed, three of the oldest and best-known of which are proportionality, envy-freeness, and equitability. An allocation is said to be \emph{proportional} if the utility that
each agent gets from the bundle she receives is at least a $1/n$ fraction of her utility for the whole set of resources, where $n$ is the number of agents among whom we divide the resources. The allocation is called \emph{envy-free} if every agent thinks that her bundle is at least as good as the bundle of any other agent, and \emph{equitable} if all agents have the same utility for their own bundle. It turns out that there is a significant distinction between the two types of resources with respect to these notions. On the one hand, when resources are divisible, allocations that satisfy the three notions simultaneously always exist \cite{Alon87}. On the other hand, a simple example with two agents who both positively value a single item already shows that the existence of a fair division cannot be guaranteed for any of the notions when we deal with indivisible items.

In this paper, we study the problem of allocating indivisible items with the added feature that the items lie on a line. We are interested in finding a fair allocation that moreover satisfies the requirement of \emph{contiguity}, i.e., the bundle that each agent receives forms a contiguous block on the line. Several practical applications fit into this model. For instance, when we divide offices between research groups on the same floor, it is desirable that each research group get a contiguous block of offices in order to facilitate communication within the group. Likewise, when we allocate retail units on a street, the retailers often prefer to have a contiguous block of units in order to operate a larger store. The contiguity condition can also be interpreted in the temporal sense, as opposed to the spatial sense described thus far. An example is a situation where various organizers wish to use the same conference center for their conferences. Not surprisingly, the organizers typically want to schedule a conference in a contiguous block of time rather than during several separate periods. 

Since allocations that satisfy any of the three fairness notions do not always exist in general, the same is necessarily true when we restrict our attention to contiguous allocations. Nevertheless, we show that in light of the contiguity requirement, the existence of allocations that satisfy approximate versions of the notions can still be guaranteed. More precisely, for each notion we define an approximate version that depends on an additive factor $\epsilon\geq 0$. An allocation is said to be $\epsilon$-proportional if the utility of each agent is at most $\epsilon$ away from her ``proportional share'', $\epsilon$-envy-free if each agent envies any other agent by at most $\epsilon$, and $\epsilon$-equitable if the utilities of any two agents differ by at most $\epsilon$. Denoting the maximum utility of an agent for an item by $u_{\text{max}}$, we establish the existence of a contiguous $u_{\text{max}}$-proportional allocation and a contiguous $u_{\text{max}}$-equitable allocation for any number of agents, a contiguous $u_{\text{max}}$-envy-free allocation for two agents, and a contiguous $2u_{\text{max}}$-envy-free allocation for any number of agents. Importantly, the approximation factors do not degrade as the number of agents or items grows. We also prove that our approximation factor is the best possible for proportionality and equitability with any number of agents as well as for envy-freeness with two agents. Finally, for proportionality the factor can be improved to $\frac{n-1}{n}\cdot u_{\text{max}}$ if we know the number $n$ of agents, and we show that this is again tight.

Our results suggest that adding the contiguity requirement does not entail extra costs in terms of the approximation guarantees. Indeed, the approximation factors for proportionality and equitability with any number of agents and for envy-freeness with two agents remain tight even if we allow arbitrary allocations. This can be seen as somewhat surprising, since the space of contiguous allocations is significantly smaller than that of arbitrary allocations. Indeed, when there are $n$ agents and $m$ items, the number of arbitrary allocations is $n^m$, while the number of contiguous allocations for a fixed order of items on a line is at most $\binom{m+n-1}{n-1}n!$. The latter quantity is much less than the former if $m$ is large compared to $n$.

In addition, we investigate the efficiency loss of contiguous allocations due to fairness constraints using the \emph{price of fairness} concept initiated by Caragiannis et al. \cite{CaragiannisKaKa12}. The price of fairness quantifies the loss of social welfare that is necessary if we impose a fairness constraint on the allocation. A low price of fairness means that we can get fairness at virtually no extra cost on social welfare, while a high price of fairness implies that even the most efficient ``fair'' allocation has social welfare far below that of the most efficient allocation overall. Caragiannis et al. studied the price of fairness for the three notions of fairness using utilitarian welfare for both divisible and indivisible items. Later, Aumann and Dombb \cite{AumannDo15} focused on contiguous allocations of divisible items and considered both utilitarian and egalitarian welfare. In this paper, we complete the picture by providing tight or almost tight bounds on the price of fairness for contiguous allocations of indivisible items, again for all three classical notions of fairness and with respect to both utilitarian and egalitarian welfare. Our results are summarized in Table~\ref{table:bigsummary} along with a comparison to results from previous work.

    \begin{table*}
\begin{center}
    \begin{tabular}{ | c || c | c | c|| c | c | }
    \hline
    \textbf{Indivisible} & \multicolumn{3}{ c|| }{Contiguous (this work)} & \multicolumn{2}{ c| }{Non-contiguous (\cite{CaragiannisKaKa12})} \\ \hline
     & \multicolumn{2}{ c| }{Utilitarian} & Egalitarian & \multicolumn{2}{ c| }{Utilitarian}  \\ \hline
     & Lower & Upper &  & Lower & Upper  \\ \hline
    Proportionality &  \multicolumn{2}{ c| }{$n-1+\frac{1}{n}$} &  1 &  \multicolumn{2}{ c| }{$n-1+\frac{1}{n}$}  \\ \hline
    Equitability &  \multicolumn{2}{ c| }{\begin{tabular}[c]{@{}c@{}}$\frac{3}{2}$ for $n=2$\\$\infty$ for $n>2$\end{tabular}}  & \begin{tabular}[c]{@{}c@{}}1 for $n=2$\\$\infty$ for $n>2$\end{tabular}  &  \multicolumn{2}{ c| }{\begin{tabular}[c]{@{}c@{}}2 for $n=2$\\$\infty$ for $n>2$\end{tabular}}  \\ \hline
    Envy-freeness &  $\frac{\lfloor\sqrt{n}\rfloor}{2}$ & $\frac{\sqrt{n}}{2}+1-o(1)$  & $\frac{n}{2}$  & $\frac{3n+7}{9}-O\left(\frac{1}{n}\right)$ &  $n-\frac{1}{2}$  \\ \hline \hline
    \textbf{Divisible} & \multicolumn{3}{ c|| }{Contiguous (\cite{AumannDo15})} & \multicolumn{2}{ c| }{Non-contiguous (\cite{CaragiannisKaKa12})} \\ \hline
     & \multicolumn{2}{ c| }{Utilitarian} & Egalitarian & \multicolumn{2}{ c| }{Utilitarian}  \\ \hline
     & Lower & Upper &  & Lower & Upper  \\ \hline
    Proportionality & $\frac{\sqrt{n}}{2}^*$  & $\frac{\sqrt{n}}{2}+1-o(1)$  &  1  & $\Omega(\sqrt{n})$  & $O(\sqrt{n})$  \\ \hline
    Equitability &  $n-1+\frac{1}{n}$ & $n$  & 1  & $\frac{(n+1)^2}{4n}$  & $n$ \\ \hline
    Envy-freeness & $\frac{\sqrt{n}}{2}^*$  &  $\frac{\sqrt{n}}{2}+1-o(1)$ & $\frac{n}{2}$  & $\Omega(\sqrt{n})$  &  $n-\frac{1}{2}$  \\ \hline
    \end{tabular}
    \vspace{5mm}
    \caption{Comparison of our results on the price of fairness to previous results in \cite{AumannDo15,CaragiannisKaKa12}. The bounds with an asterisk hold for infinitely many values of $n$.}
    \label{table:bigsummary}
\end{center}
    \end{table*}

\subsection{Related work}

The contiguity condition has been studied with respect to the three classical fairness notions in the context of divisible items, often represented by a cake, with the motivation that one wants to avoid giving an agent a ``union of crumbs''. In particular, Dubins and Spanier \cite{DubinsSp61} exhibited a moving-knife algorithm that guarantees a contiguous proportional allocation. Cechl\'{a}rov\'{a} et al. \cite{CechlarovaDoPi13} showed that for any ordering of the agents, a contiguous equitable allocation that assigns contiguous pieces to the agents in that order exists. Stromquist \cite{Stromquist80,Stromquist08} proved that a contiguous envy-free allocation always exists, but cannot be found by a finite algorithm. Su \cite{Su99} used techniques involving Sperner's lemma to establish the existence of a contiguous envy-free allocation and moreover considered the related problem of rent partitioning.  Several of the results in this paper are discrete analogs of the results in the divisible setting. For example, the algorithm in Theorem \ref{thm:proportional} is a discrete analog of the Dubins-Spanier protocol, while Theorem \ref{thm:equitable2} mirrors the analogous result in the divisible setting by Aumann and Dombb \cite{AumannDo15}.

Recently, Bouveret et al. \cite{BouveretCeEl17} studied the allocation of indivisible items on a line with the contiguity condition and showed that determining whether a contiguous fair allocation exists is NP-hard when the fairness notion considered is either proportionality or envy-freeness. They also considered a more general model of the relationship between items where the items are vertices of an undirected graph. Aumann et al. \cite{AumannDoHa13} investigated the problem of finding a contiguous allocation that maximizes welfare for both divisible and indivisible items. They showed that while it is NP-hard to find the optimal contiguous allocation, there exists an efficient algorithm that yields a constant factor approximation. Bei et al. \cite{BeiChHu12} and Cohler et al. \cite{CohlerLaPa11} also considered the objective of maximizing welfare, but under the additional fairness constraint of proportionality and envy-freeness, respectively.

Additively approximating fairness notions using $u_{\text{max}}$, the highest utility of an agent for an item, has been studied before. Lipton et al. \cite{LiptonMaMo04} showed that without the contiguity requirement, a $u_{\text{max}}$-envy-free allocation exists even for general monotone valuations. Caragiannis et al. \cite{CaragiannisKuMo16} used the term ``envy-freeness up to one good'' (EF1) to refer to a closely related property of an allocation.

Besides the allocation of goods, the price of fairness has also been investigated for the allocation of chores. In particular, Caragiannis et al. \cite{CaragiannisKaKa12}, who initiated this line of research, studied the notion for both divisible and indivisible chores. Heydrich and van Stee \cite{HeydrichVa15} likewise considered the setting of divisible chores but, similarly to our work and that of Aumann and Dombb \cite{AumannDo15}, focused on contiguous allocations. Finally, Bil\`{o} et al. \cite{BiloFaFl16} applied this concept to machine scheduling problems.

\section{Preliminaries}

Let $N=\{1,2,\dots,n\}$ denote the set of agents, and $M=\{1,2,\dots,m\}$ the set of items to be allocated. We assume that the items lie on a line in this order.

Each agent $i\in N$ has some nonnegative utility $u_i(j)$ for item $j\in M$. For an agent $i$, define $u_{i,\text{max}}:=\max_{j\in M}u_i(j)$ to be the highest utility of $i$ for an item. Let $u_{\text{max}}:=\max_{i\in N}u_{i,\text{max}}$ be the highest utility of any agent for an item. As is very common (e.g., \cite{BouveretLe16,CaragiannisKuMo16,DickersonGoKa14,ManurangsiSu17}), we assume for most of the paper that utilities are \emph{additive}. Additivity means that $u_i(M')=\sum_{j\in M'}u_i(j)$ for any agent $i$ and any subset of items $M'\subseteq M$. An \emph{allocation} $\mathcal{M}=(M_1,\dots,M_n)$ is a partition of all items into bundles for the agents so that agent $i$ receives bundle $M_i$. The \emph{utilitarian welfare} of $\mathcal{M}$ is $\sum_{i\in N}u_i(M_i)$ and the \emph{egalitarian welfare} of $\mathcal{M}$ is $\min_{i\in N}u_i(M_i)$. We call the allocation \emph{contiguous} if each bundle $M_i$ forms a contiguous block of items on the line. Furthermore, we refer to a setting with agents, items, and utility functions as an \emph{instance}.

We are now ready to define the fairness notions that we will consider in this paper. We use additive versions of approximation; this is much stronger than multiplicative versions as the number of items grows.

\begin{definition}
\label{def:proportional}
An allocation $\mathcal{M}=(M_1,\dots,M_n)$ is said to be \emph{proportional} if $u_i(M_i)\geq\frac{1}{n}\cdot u_i(M)$ for all $i\in N$. For $\epsilon\geq 0$, the allocation is said to be \emph{$\epsilon$-proportional} if $u_i(M_i)\geq\frac{1}{n}\cdot u_i(M)-\epsilon$ for all $i\in N$. We refer to $\frac{1}{n}\cdot u_i(M)$ as the \emph{proportional share} of agent $i$.
\end{definition}

\begin{definition}
\label{def:envyfree}
An allocation $\mathcal{M}=(M_1,\dots,M_n)$ is said to be \emph{envy-free} if $u_i(M_i)\geq u_i(M_j)$ for all $i,j\in N$. For $\epsilon\geq 0$, the allocation is said to be \emph{$\epsilon$-envy-free} if $u_i(M_i)\geq u_i(M_j)-\epsilon$ for all $i,j\in N$.
\end{definition}

\begin{definition}
\label{def:equitable}
An allocation $\mathcal{M}=(M_1,\dots,M_n)$ is said to be \emph{equitable} if $u_i(M_i)=u_j(M_j)$ for all $i,j\in N$. For $\epsilon\geq 0$, the allocation is said to be \emph{$\epsilon$-equitable} if $|u_i(M_i)-u_j(M_j)|\leq\epsilon$ for all $i,j\in N$.
\end{definition}

There is a strong relation between proportionality and envy-freeness, as the following proposition shows. 

\begin{proposition}
\label{prop:EFproportional}
Any $\epsilon$-envy-free allocation is $\epsilon$-proportional.
\end{proposition}

\begin{proof}
Consider an $\epsilon$-envy-free allocation $\mathcal{M}$, and fix an agent $i$. We have $u_i(M_i)\geq u_i(M_j)-\epsilon$ for all $j\in N$. Summing the $n$ inequalities and dividing by $n$, we get $$u_i(M_i)\geq\frac{1}{n}\sum_{j=1}^n \left(u_i(M_j)-\epsilon\right)=\frac{1}{n}\cdot u_i(M)-\epsilon.$$ Hence the allocation is $\epsilon$-proportional. 
\end{proof}

In particular, when $\epsilon=0$, the proposition reduces to the well-known fact that any envy-free allocation is proportional. When there are two agents, proportional allocations are also envy-free (and in fact, more generally, $\epsilon$-proportional allocations are $2\epsilon$-envy-free.) This is, however, not necessarily the case if there are at least three agents. An example is when an agent values her own bundle $1/n$ of the whole set of items and values the bundle of another agent the remaining $(n-1)/n$ of the whole set of items. On the other hand, equitability neither implies nor is implied by proportionality or envy-freeness.

For each of the three (non-approximate) fairness notions defined above, the set of instances for which a contiguous allocation satisfying the notion exists is strictly smaller than the corresponding set when contiguity is not required. Indeed, suppose that there are three items and two agents who share a common utility function $u$ with $u(1)=u(3)=1$ and $u(2)=2$. An allocation in which one agent gets items 1 and 3 while the other agent gets item 2 is proportional, envy-free, and equitable. In contrast, no contiguous allocation satisfies any of the three properties.

We end this section by giving the definition of the various forms of the price of fairness. Following Caragiannis et al. \cite{CaragiannisKaKa12}, we assume the normalization $u_i(M)=1$ for all $i\in N$ when we consider these price of fairness notions.

\begin{definition}
Given an instance (along with a set of allocations considered), its \emph{utilitarian price of proportionality} (resp., \emph{utilitarian price of equitability}, \emph{utilitarian price of envy-freeness}) is defined as the ratio of the utilitarian welfare of the optimal allocation over the utilitarian welfare of the best proportional (resp., equitable, envy-free) allocation. If a proportional (resp., equitable, envy-free) allocation does not exist, the utilitarian price of proportionality (resp., equitability, envy-freeness) is not defined for that instance. The (overall) \emph{utilitarian price of proportionality} (resp., \emph{utilitarian price of equitability}, \emph{utilitarian price of envy-freeness}) is then the supremum utilitarian price of proportionality (resp., utilitarian price of equitability, utilitarian price of envy-freeness) over all instances.

The \emph{egalitarian price of proportionality}, \emph{egalitarian price of equitability}, and \emph{egalitarian price of envy-freeness} are defined analogously.
\end{definition}

\section{Proportionality}

We begin with proportionality. Our first result shows the existence of a contiguous allocation in which every agent receives at least her proportional share minus $\frac{n-1}{n}$ times her utility for her highest-valued item.

\begin{theorem}
\label{thm:proportional}
Given any instance, there exists a contiguous allocation $\mathcal{M}$ such that $$u_i(M_i)\geq\frac{1}{n}\cdot u_i(M)-\frac{n-1}{n}\cdot u_{i,\text{max}}$$ for all agents $i\in N$. In particular, there exists a contiguous $\frac{n-1}{n}\cdot u_{\text{max}}$-proportional allocation.
\end{theorem}

\begin{proof}
We process the items from left to right using the following algorithm.
\begin{enumerate}
\item Set the current block to the empty block.
\item If the current block yields utility at least $\frac{1}{n}\cdot u_i(M)-\frac{n-1}{n}\cdot u_{i,\text{max}}$ to some agent $i$, give the block to the agent. (If several agents satisfy this condition, choose one arbitrarily.)
\begin{itemize}
\item If all agents have received a block as a result of this, allocate the leftover items arbitrarily and terminate. 
\item Otherwise, if some agent receives a block in this step, remove that agent from consideration and return to Step 1.
\end{itemize}
\item Add the next item to the current block and return to Step 2.
\end{enumerate}

If an agent $i$ receives a block of items from this algorithm, she obtains utility at least $\frac{1}{n}\cdot u_i(M)-\frac{n-1}{n}\cdot u_{i,\text{max}}$. Hence it suffices to show that the algorithm allocates a block to every agent. To this end, we show by (backward) induction that when there are $k$ agents who have not been allocated a block, each agent $i$ among them has utility at least $\frac{k}{n}\cdot u_i(M)-\frac{n-k}{n}\cdot u_{i,\text{max}}$ for the remaining items. This will imply that the last agent has utility at least $\frac{1}{n}\cdot u_i(M)-\frac{n-1}{n}\cdot u_{i,\text{max}}$ left, which is enough to satisfy our condition. 

The base case $k=n$ trivially holds. For the inductive step, assume that the statement holds when there are $k+1$ agents left, and consider an agent $i$ who is \emph{not} the next one to receive a block. When there are $k+1$ agents left, her utility for the remaining items is at least $\frac{k+1}{n}\cdot u_i(M)-\frac{n-k-1}{n}\cdot u_{i,\text{max}}$. Since she does not receive the next block, her utility for the block excluding its last item is less than $\frac{1}{n}\cdot u_i(M)-\frac{n-1}{n}\cdot u_{i,\text{max}}$. This means that her utility for the block is less than $\frac{1}{n}\cdot u_i(M)+\frac{1}{n}\cdot u_{i,\text{max}}$. Hence her utility for the remaining items is at least 
$\left(\frac{k+1}{n}\cdot u_i(M)-\frac{n-k-1}{n}\cdot u_{i,\text{max}}\right)-\left(\frac{1}{n}\cdot u_i(M)+\frac{1}{n}\cdot u_{i,\text{max}}\right)$, which is equal to $\frac{k}{n}\cdot u_i(M)-\frac{n-k}{n}\cdot u_{i,\text{max}}$,
as desired. 
\end{proof}

The algorithm in Theorem \ref{thm:proportional} is similar to the Dubins-Spanier algorithm for proportional cake-cutting \cite{DubinsSp61} and runs in time $\Theta(mn)$, which is the best possible since the input also has size $\Theta(mn)$. It can also be implemented as a mechanism that does not elicit the full utility functions from the agents, but instead asks them to indicate when the value of the current block reaches their threshold. While the mechanism is not truthful,\footnote{Indeed, an agent can profit by not claiming a block when her threshold is reached if she believes that no other agent is close to their own threshold yet.}
a truthful agent always receives no less than her proportional share minus $\frac{n-1}{n}$ times her utility for the item she values most.

As the next example shows, the additive approximation factor $\frac{n-1}{n}\cdot u_{\text{max}}$ is the best possible in the sense that the existence of a contiguous $\alpha\cdot\frac{n-1}{n}\cdot u_{\text{max}}$-proportional allocation is not guaranteed for any $\alpha<1$. In fact, this is the case even if we remove the contiguity requirement.

\begin{example}
\label{ex:proptight}
Suppose that there are $m=n-1$ items any of which each agent has a utility of 1. The proportional share of every agent is $\frac{n-1}{n}$. On the other hand, in any (not necessarily contiguous) allocation, some agent does not receive an item and therefore has a utility of 0. For any fixed $\alpha<1$, the utility of this agent is less than her proportional share minus $\alpha\cdot\frac{n-1}{n}\cdot u_{\text{max}}$.
\end{example}

Even though a contiguous $u_{\text{max}}$-proportional allocation always exists, in some cases we might also want to choose the ordering on the line in which the agents are allocated blocks of items, in addition to imposing the contiguity requirement. For instance, the owner of a conference center could have a preferred lineup of the conferences, and a building manager might want to assign offices in certain parts of the floor to certain research groups. Nevertheless, the following example shows that for approximate proportionality, the ordering cannot always be chosen arbitrarily.

\begin{example}
\label{ex:proporder}
Suppose that there are two agents and $m\geq 6$ items. The first agent has utility 1 for the last three items and 0 otherwise, while the second agent has utility 1 for every item. The proportional share of the two agents less half of the utility for their highest-valued item is 1 and $\frac{m-1}{2}$, respectively. If we want to give a left block to the first agent and the remaining right block to the second agent, the left block needs to include up to item $m-2$. But this means that the second agent gets utility at most 2, which is less than the required $\frac{m-1}{2}$.
\end{example}

By increasing $m$ and the number of items for which the first agent has utility 1, we can extend the example to show that the existence of a contiguous allocation with a fixed ordering of agents is not guaranteed even if we weaken the approximation factor $u_{\text{max}}$ to $ku_{\text{max}}$ for any $k>1$.

\section{Equitability}

We next consider equitability. As with proportionality, we show that a contiguous allocation in which the values of different agents for their own block differ by no more than $u_{\text{max}}$ always exists. Unlike for proportionality, however, for equitability we can additionally choose the order in which the agents receive blocks on the line.

\begin{theorem}
\label{thm:equitable}
Given any instance and any ordering of the agents, there exists a contiguous $u_{\text{max}}$-equitable allocation in which the agents are allocated blocks of items on the line according to the ordering.
\end{theorem}

Note that in order to ensure that all agents are treated equally, one should normalize the utilities across agents before applying Theorem~\ref{thm:equitable}, for example by rescaling the utilities so that $u_i(M)=1$ for all $i\in N$.

\begin{proof}
Assume without loss of generality that the required ordering of agents is $1,2,\dots,n$ from left to right. Start with an arbitrary allocation  satisfying the ordering. For any allocation $\mathcal{M}$, let $\max(\mathcal{M})=\max_{i=1}^n u_i(M_i)$ and $\min(\mathcal{M})=\min_{i=1}^n u_i(M_i)$. In each iteration, as long as the allocation does not satisfy the approximate equitability, we will move an item at the end of a block to the block that the item is adjacent to. Here is the description of the algorithm.\footnote{The algorithm is inspired by work on block partitions of sequences \cite{BaranyGr15}.}

\begin{enumerate}
\item Choose a block $M_i$ such that $u_i(M_i)=\max(\mathcal{M})$. If there are many such blocks, choose one arbitrarily.
\item If $\max(\mathcal{M})\leq\min(\mathcal{M})+u_{\text{max}}$, stop and return the current allocation.
\item Choose a block $M_j$ such that $u_j(M_j)=\min(\mathcal{M})$. If there are many such blocks, choose arbitrarily from the ones that minimize $|j-i|$.
\item Let $M_k$ be the block between $M_i$ and $M_j$ that is next to $M_j$, i.e., $k=j-1$ if $j>i$ and $k=j+1$ if $j<i$. (It is possible that $k=i$.) The block $M_k$ must be non-empty; otherwise we would have chosen $M_k$ instead of $M_j$. Move the item in $M_k$ that is adjacent to $M_j$ to $M_j$. 
\begin{enumerate}
\item If $k=i$ and the moved item has nonzero utility for agent $i$, go to Step 1.
\item Else, go to Step 2.
\end{enumerate}
\end{enumerate}
If the algorithm terminates, then as the ordering of the agents for the blocks never changed, the algorithm returns a $u_{\text{max}}$-equitable allocation with the desired ordering. Hence it remains to show that the algorithm terminates.

To this end, observe that when an item is moved in Step 4 of the algorithm, no new block with utility $\max(\mathcal{M})$ or more to the block owner is created. Indeed, the block that gets an additional item now yields utility at most $\min(\mathcal{M})+u_{\text{max}}$ to its owner, which is less than $\max(\mathcal{M})$ since the condition in Step 2 is not yet satisfied. Moreover, since items are only being moved farther away from the main block $M_i$, we must eventually reach a point where $k=i$; formally, the quantity $\sum_{z=1}^n|i-z||M_z|$ strictly increases, where $|M_z|$ denotes the number of items in $M_z$. Since the quantity is bounded from above, after a finite number of moves we will have $k=i$ and meet the condition of Step 4(a).

The argument in the previous paragraph shows that the number of blocks with utility $\max(\mathcal{M})$ decreases during the course of the algorithm. When this number reaches zero, the value $\max(\mathcal{M})$ decreases. Since there are only a finite number of allocations, the algorithm must terminate, as claimed. 
\end{proof}

The algorithm in Theorem~\ref{thm:equitable} runs in time $O(n^2m^4)$. For each iteration, computing the maximum and minimum blocks takes $O(m)$. There are $O(m^2)$ possible blocks. Each block cannot be used as the block $M_i$ in Step 1 more than once for each of the $n$ agents, since no new blocks with utility $\max(\mathcal{M})$ is created during an execution of the algorithm. Finally, once the block $M_i$ is fixed, the quantity $\sum_{z=1}^n|i-z||M_z|$ can increase at most $O(mn)$ times, yielding the claimed running time.

Example \ref{ex:proptight} shows that for any number of agents, the approximation factor $u_{\text{max}}$ for equitability cannot be improved even if we remove the contiguity requirement (and hence also the ordering). On the other hand, using the same algorithm, we can generalize Theorem \ref{thm:equitable} to any monotonic, not necessarily additive utility function with zero utility for the empty set. In particular, a $u_{\text{max}}$-equitable allocation can be found when agents are endowed with such utility functions, where the generalized definition of $u_{\text{max}}$ is the highest marginal utility of any agent for a single item, i.e., $u_{\text{max}}=\max_{i\in N,j\in M,S\subseteq M}(u_i(S\cup\{j\})-u_i(S))$.

Although the algorithm in Theorem \ref{thm:equitable} guarantees that an approximate equitable allocation exists, such an allocation can be ``equally bad'' rather than ``equally good'' for the agents. Indeed, if we start with an allocation that yields zero utility to every agent, then the algorithm will terminate immediately despite the possible existence of an equitable allocation with positive utility for all agents. If we insist on choosing the ordering of the agents, then the next example shows that a situation that leaves some agent unhappy may be unavoidable.

\begin{example}
Suppose that there are two items and two agents with $u_1(2)=u_2(1)=1$ and $u_1(1)=u_2(2)=0$. The allocation that gives item 1 to agent 2 and item 2 to agent 1 yields a utility of 1 to both agents. If we require that agent 1 receive a left block and agent 2 a right block, however, some agent is necessarily left with no utility.
\end{example}

Nevertheless, we show next that if we allow the freedom of choosing the ordering of the agents, then an allocation with a better efficiency guarantee for the agents can always be found. In particular, we can find an allocation whose egalitarian welfare equals the highest egalitarian welfare over all contiguous allocations of the instance. The proof mirrors that of the analogous result for divisible items by Aumann and Dombb \cite{AumannDo15}.

\begin{theorem}
\label{thm:equitable2}
Given any instance, there exists a contiguous $u_{\text{max}}$-equitable allocation whose egalitarian welfare equals the highest egalitarian welfare over all contiguous allocations of the instance.
\end{theorem}

\begin{proof}
Denote by $w$ the highest egalitarian welfare over all contiguous allocations of the instance, and let $\mathcal{M}=(M_1,\dots,M_n)$ be an allocation achieving this welfare, where we assume without loss of generality that $M_1,\dots,M_n$ lie in this order on the line. We claim that there exists a contiguous allocation in which each agent receives utility in the range $[w,w+u_{\text{max}}]$. To prove this claim, we proceed by induction on the number of agents. The base case $n=1$ is trivial.

Suppose that the claim holds for $n-1$ agents. Starting from the allocation $\mathcal{M}$, we will move the boundaries between consecutive blocks to the left, beginning with the boundary between the first two blocks and continuing rightwards. In particular, for $1\leq i\leq n-1$, we move the boundary between blocks $i$ and $i+1$ to the left until $w\leq u_i(M_i)<w+u_{\text{max}}$. Note that this is always possible since each item is worth no more than $u_{\text{max}}$ to any agent. After moving all $n-1$ boundaries, we reach an allocation where $w\leq u_i(M_i)<w+u_{\text{max}}$ for $1\leq i\leq n-1$ and $u_n(M_n)\geq w$. If $u_n(M_n)\leq w+u_{\text{max}}$, our claim is proved. Hence we can assume that $u_n(M_n)> w+u_{\text{max}}$. 

Next, we move the boundary between blocks $n-1$ and $n$ to the right until $w< u_n(M_n)\leq w+u_{\text{max}}$; this is possible for similar reasons as above. Temporarily remove agent $n$ and her block. The current allocation of the remaining items to the $n-1$ agents still yields utility at least $w$ to each agent. On the other hand, there is no contiguous allocation of these items to the $n-1$ agents with egalitarian welfare more than $w$. Indeed, such an allocation would imply the existence of a contiguous allocation of all items to the $n$ agents with egalitarian welfare strictly more than $w$. This means that the current distribution of the remaining items to the $n-1$ agents maximizes the egalitarian welfare, and that welfare is $w$. By the inductive hypothesis, there exists a contiguous allocation of the remaining items to the $n-1$ agents so that the utility of each agent is in the range $[w,w+u_{\text{max}}]$.\footnote{In fact, the inductive hypothesis implies the existence of such an allocation in which the utility of each agent is in $[w,w+u_{-n,\text{max}}]$, where $u_{-n,\text{max}}\leq u_{\text{max}}$ denotes the highest utility of any of the first $n-1$ agents for an item.} Combined with $w<u_n(M_n)\leq w+u_{\text{max}}$, we have our claim. 
\end{proof}

As is the case for Theorem \ref{thm:equitable}, the result can be generalized to monotonic, not necessarily additive utility functions with zero utility for the empty set, where $u_{\text{max}}$ is again defined as the highest marginal utility of any agent for a single item. A similar argument also shows that for any ordering of the agents, there exists a contiguous $u_{\text{max}}$-equitable allocation whose egalitarian welfare equals the highest egalitarian welfare over all contiguous allocations \emph{with that ordering of the agents} for the instance. However, the proof of Theorem \ref{thm:equitable2} does not give rise to an efficient algorithm for computing a desired allocation.

\section{Envy-freeness}

We now turn to envy-freeness. If we remove the contiguity requirement, it is well-known that a simple algorithm yields a $u_{\text{max}}$-envy-free allocation for any number of agents and items: Let the agents pick their favorite item in a round-robin manner from the remaining items until all items are allocated.
We show in this section that a $u_{\text{max}}$-envy-free allocation exists when there are two agents, and a $2u_{\text{max}}$-envy-free allocation exists for an arbitrary number of agents.

For two agents, Theorem \ref{thm:proportional} directly implies the following.

\begin{theorem}
\label{thm:envyfreetwo}
Given any instance with two agents, there exists a contiguous allocation such that agent $i$ has envy at most $u_{i,\text{max}}$ toward the other agent. In particular, there exists a contiguous $u_{\text{max}}$-envy-free allocation.
\end{theorem}

Tightness of the approximation factor $u_{\text{max}}$ follows from Example \ref{ex:proptight} with $n=2$. Moreover, an example similar to Example \ref{ex:proporder} shows that the result does not hold if we fix the ordering of the agents, even when we replace $u_{\text{max}}$ by $ku_{\text{max}}$ for some $k>1$.

To tackle the general setting with an arbitrary number of agents, we model the items as divisible items. Since a contiguous envy-free allocation always exists for divisible items \cite{Stromquist80}, we can round such an allocation to obtain an approximate envy-free allocation for indivisible items.

\begin{theorem}
\label{thm:envyfreegeneral}
Given any instance, there exists a contiguous allocation such that agent $i$ has envy less than $2u_{i,\text{max}}$ toward any other agent. In particular, there exists a contiguous $2u_{\text{max}}$-envy-free allocation.
\end{theorem}

\begin{proof}
Consider a cake represented by the interval $[0,m]$. For $j\in M$, agent $i$ has uniform utility $u_i(j)$ for the interval $[j-1,j]$. Take any contiguous envy-free allocation of the cake. We round the allocation as follows: For each item $j$, if point $j$ is in the interior of a piece, allocate the item to the agent who owns that piece. Else, point $j$ is at the boundary between two pieces, and we allocate item $j$ to the agent who owns the piece to the left.

The resulting allocation is contiguous; we show that each agent $i$ has envy less than $2u_{i,\text{max}}$. The agent has no envy before the rounding. As a result of the rounding, she loses utility less than $u_{i,\text{max}}$, and any other agent gains utility less than $u_{i,\text{max}}$ from her point of view. Hence agent $i$ has envy less than $2u_{i,\text{max}}$, as claimed. 
\end{proof}

The guarantee in Theorem \ref{thm:envyfreegeneral} can be strengthened if agents have \emph{binary} utilities, i.e., the utility of each agent $i$ for any item is either $x_i$ or 0, for some $x_i>0$. In other words, each agent either approves or disapproves each item. Such utilities have been considered in the literature \cite{BouveretLe16}.

\begin{corollary}
Given any instance with agents having binary utilities, there exists a contiguous allocation such that agent $i$ has envy at most $u_{i,\text{max}}$ toward any other agent.
\end{corollary}

\begin{proof}
We do not need to consider any agent who does not approve any item, so we may assume that $u_{i,\text{max}}=x_i>0$ for each agent $i$. Theorem \ref{thm:envyfreegeneral} guarantees the existence of an allocation such that agent $i$ has envy less than $2x_i$ toward any other agent. However, since the utility of agent $i$ for any set of items is an integer multiple of $x_i$, the agent's envy can be at most $x_i$.
\end{proof}


\section{Price of Fairness}

In this section, we quantify the price of fairness for contiguous allocations of indivisible items with respect to the three notions of fairness and the two types of welfare. We derive tight or almost tight bounds for each of the six resulting combinations. Previous work has studied the problem for the setting of arbitrary (i.e., not necessarily contiguous) allocations of divisible and indivisible items \cite{CaragiannisKaKa12} as well as contiguous allocations of divisible items \cite{AumannDo15}; our results therefore close the remaining gap. The comparison of our results to previous work is shown in Table~\ref{table:bigsummary}. In fact, for several of the results we will be able to adjust arguments from the previous work mentioned to our setting. 

\begin{theorem}
\label{thm:utilproportional}
The utilitarian price of proportionality for contiguous allocations of indivisible items is $n-1+\frac{1}{n}$.
\end{theorem}

\begin{proof}
\emph{Upper bound}: Consider an arbitrary instance. If a contiguous allocation with maximum utilitarian welfare of the instance is also proportional, the price of proportionality is 1. Else, some agent has utility less than $\frac{1}{n}$, and so the utilitarian welfare of this allocation is less than $n-1+\frac{1}{n}$. On the other hand, in any proportional allocation, every agent has utility at least $\frac{1}{n}$ and therefore the utilitarian welfare is at least 1.

\emph{Lower bound}: Let $m=2n-1$, $0<\epsilon<\frac{1}{n}$, and assume that the utilities are as follows:

\begin{itemize}
\item For $i=1,\dots,n-1$: $u_i(2i-1)=\epsilon$, $u_i(2i)=\frac{1}{n}-\epsilon$, $u_i(2i+1)=\frac{n-1}{n}$, and $u_i(j)=0$ otherwise.
\item $u_n(2j-1)=\frac{1}{n}-\epsilon$ for $j=1,\dots,n-1$, $u_n(2n-1)=\frac{1}{n}+(n-1)\epsilon$, and $u_n(j)=0$ otherwise.
\end{itemize}

Consider the contiguous allocation that assigns the first item to agent $n$ and items $2i$ and $2i+1$ to agent $i$ for $i=1,\dots,n-1$. The utilitarian welfare of this allocation is $\left(\frac{1}{n}-\epsilon\right)+(n-1)(1-\epsilon)=n-1+\frac{1}{n}-n\epsilon$.

On the other hand, consider any proportional allocation. Each agent must get at least one odd-numbered item in order for her utility to be at least $\frac{1}{n}$. Since there are $n$ odd-numbered items, every agent must get exactly one such item. It is clear that agent $n$ must get item $2n-1$. Given that, agent $n-1$ must get both items $2n-3$ and $2n-2$. Applying this argument repeatedly, we find that agent $i$ must get items $2i-1$ and $2i$. The utilitarian welfare of the resulting allocation, which is indeed proportional, is $1+(n-1)\epsilon$. Taking $\epsilon\rightarrow 0$, we have our claim. 
\end{proof}

\begin{theorem}
\label{thm:utilequitable}
The utilitarian price of equitability for contiguous allocations of indivisible items is $\frac{3}{2}$ for $n=2$ and infinite for $n>2$.
\end{theorem}

\begin{proof}
\emph{$n=2$}: Consider an arbitrary instance. In an equitable contiguous allocation of the instance with maximum utilitarian welfare, both agents must have utility at least $\frac{1}{2}$; otherwise they can switch their bundles. Let $x\geq\frac{1}{2}$ denote their utility in this allocation, and assume without loss of generality that the first agent gets a left block and the second agent the remaining right block. Consider any contiguous allocation. If the first agent gets a left block, then at least one of the agents gets utility at most $x$. Similarly, if the first agent gets a right block, then at least one of the agents gets utility at most $1-x\leq x$. Hence the utilitarian welfare of any contiguous allocation of this instance is at most $x+1$. The ratio of the maximum utilitarian welfare of any contiguous allocation to that of an equitable contiguous allocation is at most $\frac{x+1}{2x}\leq\frac{3}{2}$.

To show that the bound $\frac{3}{2}$ is tight, consider the following instance with $m=4$. Let $u_1(1)=u_1(3)=u_2(2)=u_2(4)=\frac{1}{2}$ and $u_i(j)=0$ otherwise. One can check that the maximum utilitarian welfare of a contiguous allocation is $\frac{3}{2}$, whereas that of an equitable contiguous allocation is $1$.

\emph{$n>2$}: Let $m=n$, $\epsilon<\frac{1}{2}$, and assume that the utilities are as follows:

\begin{itemize}
\item For $i=1,\dots,n-1$: $u_i(i)=\epsilon$, $u_i(i+1)=1-\epsilon$, and $u_i(j)=0$ otherwise.
\item $u_n(1)=1-2\epsilon$, $u_n(n-1)=u_n(n)=\epsilon$, and $u_n(j)=0$ otherwise.
\end{itemize}

Consider the contiguous allocation that assigns the first item to agent $n$ and item $i+1$ to agent $i$ for $i=1,\dots,n-1$. The utilitarian welfare of this allocation is $(1-2\epsilon)+(n-1)(1-\epsilon)=n-(n+1)\epsilon$. On the other hand, the maximum utilitarian welfare of any equitable allocation is $n\epsilon$. Taking $\epsilon\rightarrow 0$, we have our claim. 
\end{proof}

\begin{theorem}
\label{thm:utilenvyfree}
The utilitarian price of envy-freeness for contiguous allocations of indivisible items is in the interval $\left(\frac{\lfloor\sqrt{n}\rfloor}{2},\frac{\sqrt{n}}{2}+1-o(1)\right)$.
\end{theorem}

\begin{proof}
\emph{Upper bound}: Given an arbitrary instance, consider a cake represented by the interval $[0,m]$. For $j\in M$, agent $i$ has uniform utility $u_i(j)$ for the interval $[j-1,j]$. Aumann and Dombb \cite[Theorem~2.1]{AumannDo15} showed that in this setting, the ratio of the utilitarian welfare of any contiguous allocation to that of any envy-free contiguous allocation is at most $\frac{\sqrt{n}}{2}+1-\frac{n}{4n^2-4n+2\sqrt{n}}=\frac{\sqrt{n}}{2}+1-o(1)$. Since contiguous allocations of indivisible items can be viewed as contiguous allocations of the cake, the result holds in our setting as well.

\emph{Lower bound}: Let $n=m$, $r=\lfloor\sqrt{n}\rfloor$, and assume that the utilities are as follows:

\begin{itemize}
\item For $i=1,\dots,r-1$: $u_i(ir-j)=\frac{1}{r}$ for $j=0,\dots,r-1$, and $u_i(j)=0$ otherwise.
\item $u_r(j)=\frac{1}{n-r(r-1)}$ for $j=r(r-1)+1,\dots,n$.
\item For $i=r+1,\dots,n$: $u_i(j)=\frac{1}{n}$ for all $j$.
\end{itemize}

Consider the contiguous allocation that assigns items $ir-r+1,\dots,ir$ to agent $i$ for $i=1,\dots,r-1$ and the remaining items to agent $r$. The utilitarian welfare of this allocation is $r$. On the other hand, an envy-free allocation exists, and in any such allocation, every agent gets exactly one item. Hence the utilitarian welfare is at most $2-\frac{r}{n}$, and the utilitarian price of envy-freeness is therefore at least $\frac{r}{2-\frac{r}{n}}>\frac{r}{2}=\frac{\lfloor\sqrt{n}\rfloor}{2}$. 
\end{proof}

\begin{theorem}
\label{thm:egalproportional}
The egalitarian price of proportionality for contiguous allocations of indivisible items is 1.
\end{theorem}

\begin{proof}
For an arbitrary instance, every agent has utility at least $\frac{1}{n}$ in a proportional contiguous allocation. This implies that in a proportional contiguous allocation with maximum egalitarian welfare of the instance, every agent also has utility at least $\frac{1}{n}$. Hence the latter allocation is also proportional, and the claim is proved. 
\end{proof}

\begin{theorem}
\label{thm:egalequitable}
The egalitarian price of equitability for contiguous allocations of indivisible items is $1$ for $n=2$ and infinite for $n>2$.
\end{theorem}

\begin{proof}
\emph{$n=2$}: Consider an arbitrary instance. In an equitable contiguous allocation of the instance with maximum egalitarian welfare, both agents must have utility at least $\frac{1}{2}$; otherwise they can switch their bundles. Let $x\geq\frac{1}{2}$ denote their utility in this allocation, and assume without loss of generality that the first agent gets a left block and the second agent the remaining right block. Consider any equitable allocation. If the first agent gets a left block, then at least one of the agents gets utility at most $x$. Similarly, if the first agent gets a right block, then at least one of the agents gets utility at most $1-x\leq x$. Hence the egalitarian welfare of any contiguous allocation of this instance is at most $x$ as well.

\emph{$n>2$}: We use the same example as in Theorem~\ref{thm:utilequitable}. There exists a contiguous allocation with egalitarian welfare $1-2\epsilon$, while the maximum egalitarian welfare of any equitable allocation is $\epsilon$. Taking $\epsilon\rightarrow 0$, we have our claim. 
\end{proof}

\begin{theorem}
\label{thm:egalenvyfree}
The egalitarian price of envy-freeness for contiguous allocations of indivisible items is $\frac{n}{2}$.
\end{theorem}

\begin{proof}
\emph{Upper bound}: Consider an arbitrary instance. If the maximum egalitarian welfare of a contiguous allocation of the instance is at least $\frac{1}{2}$, then the allocation is also envy-free, and the price of envy-freeness is 1. Else, the maximum egalitarian welfare of a contiguous allocation is less than $\frac{1}{2}$. On the other hand, in any envy-free allocation, each agent has a utility of at least $\frac{1}{n}$, and therefore the egalitarian welfare is at least $\frac{1}{n}$.

\emph{Lower bound}: Let $m=2n$, $0<\epsilon<\frac{1}{2n}$, and assume that the utilities are as follows:

\begin{itemize}
\item For $i=1,\dots,n-1$: $u_i(2i)=\frac{1}{2}+\epsilon$, $u_i(2n-2i-1)=\frac{1}{2}-\epsilon$, and $u_i(j)=0$ otherwise.
\item $u_n(j)=\frac{1}{2n}$ for all $j$.
\end{itemize}

Consider the contiguous allocation that assigns item $2i$ to agent $i$ for $i=1,\dots,\lfloor\frac{n-1}{2}\rfloor$, item $2n-2i-1$ to agent $i$ for $i=\lfloor\frac{n-1}{2}\rfloor+1,\dots,n-1$, and items $n,n+1,\dots,2n$ to agent $n$. The egalitarian welfare of this allocation is $\frac{1}{2}-\epsilon$.

On the other hand, there exists an envy-free allocation, namely the allocation that assigns items $2i-1$ and $2i$ to agent $i$ for all $i\in N$. Consider any envy-free allocation. If agent $n$ gets at least three items in this allocation, her bundle will include an item for which some other agent has utility $\frac{1}{2}+\epsilon$, and the allocation cannot be envy-free. Hence agent $n$ has utility at most $\frac{1}{n}$, and the egalitarian welfare is also at most $\frac{1}{n}$. Taking $\epsilon\rightarrow 0$, we have our claim. 
\end{proof}

Note that even though envy-free allocations are always proportional, the utilitarian price of envy-freeness is lower than the utilitarian price of proportionality. This is not a contradiction since we only consider instances for which a contiguous allocation satisfying the fairness notion in question exists when computing the price of fairness. Indeed, in the instance used to show the lower bound in Theorem~\ref{thm:utilproportional}, no envy-free allocation exists.

\section{Conclusion and Future Work}

In this paper, we study the problem of fairly allocating indivisible items on a line in such a way that each agent receives a contiguous block of items. This can be used to model a variety of practical situations, including allocating offices to research groups, retail units to retailers, and time slots for using a conference center to conference organizers. We show that we can find contiguous allocations that satisfy approximate versions of classical fairness notions. Notably, these approximation guarantees do not degrade as the number of agents or items grows. We also quantify the loss of efficiency that occurs when we impose fairness constraints on contiguous allocations.

We conclude the paper by presenting some directions for future work.

\begin{itemize}
\item For envy-freeness with an arbitrary number of agents, can we close the approximation factor gap between $u_{\text{max}}$ and $2u_{\text{max}}$? Can we obtain similar guarantees if we also require Pareto optimality?
\item Can we show the asymptotic existence or non-existence of contiguous allocations satisfying proportionality or envy-freeness if we assume that the utilities are drawn from certain distributions? This has been shown for non-contiguous allocations \cite{DickersonGoKa14,ManurangsiSu17,Suksompong16}.
\item Does there exist an efficient algorithm that computes an approximate equitable allocation with a nontrivial welfare guarantee?
\item How do the prices of fairness change if we define them with respect to approximate fair allocations (which always exist) instead of non-approximate fair allocations (which do not always exist)?
\end{itemize}

\bibliographystyle{abbrv}
\bibliography{main}

\end{document}